\documentclass{article}

\usepackage{amsmath, amssymb, amsthm}
\usepackage{url}
\usepackage{color}
\usepackage{verbatim}
\usepackage{authblk}

\mathchardef\mhyphen="2D

\newtheorem{thm}{Theorem}

\newtheorem{lem}[thm]{Lemma}
\newtheorem{cor}[thm]{Corollary}

\theoremstyle{definition}
\newtheorem{defn}[thm]{Definition}
\newtheorem{ex}[thm]{Example}

\theoremstyle{remark}

\newcommand{\R}{\mathcal{R}}

\newcommand{\NP}{\mathbf{NP}}
\newcommand{\REG}{\mathbf{REG}}
\newcommand{\IL}{\mathbf{IL}}
\newcommand{\LIL}{\mathbf{LIL}}
\newcommand{\CFL}{\mathbf{CFL}}
\newcommand{\CSL}{\mathbf{CSL}}

\newcommand{\ins}{\leftarrow}
\newcommand{\eps}{\epsilon}
\begin{document}

\title{Deleting Powers in Words}
\author{John Machacek}
\affil{Department of Mathematics, Michigan State University\\ \texttt{machace5@math.msu.edu}}
\date{}

\maketitle

\begin{abstract}
We consider the language consisting of all words such that it is possible to obtain the empty word by iteratively deleting powers.
It turns out that in the case of deleting squares in binary words this language is regular, and in the case of deleting squares in words over a larger alphabet the language is not regular.
However, for deleting squares over any alphabet we find that this language can be generated by a linear index grammar which is a mildly context sensitive grammar formalism.
In the general case we show that this language is generated by an indexed grammar.

\end{abstract}

\section{Introduction}

Let $\Sigma$ be a finite alphabet.
For a word $w \in \Sigma^*$ and $a \in \Sigma$ we let $|w|$ denote the length of $w$ and $|w|_a$ denote the number of occurrences of the letter $a$.
For an integer $p > 0,$ a \emph{$p$th-power} is a $p$-fold repetition $u^p = uu\cdots u$ of non-empty word $u \in \Sigma^*$.
As an example, $(ab)^3 = ababab$ is a $3$rd-power.
Given a word $w = a_1a_2\cdots a_n \in \Sigma^*,$ we say $w$ contains the word $u \in \Sigma^*$ if $u = a_i a_{i+1} \cdots a_j$ for some $1 \leq i \leq j \leq n$.
A word is called \emph{$p$th-power-free} if it contains no $p$th-powers.
For $p=2$ and $p=3$ we will refer to $p$th-powers as \emph{squares} and \emph{cubes} respectively.
We let $\Sigma_k$ denote an alphabet of size $k$, and we typically consider $\Sigma_{k} \subsetneq \Sigma_{k+1}$.
Also, for $k=2$ and $k=3$ we call elements of $\Sigma^*_k$ \emph{binary} and \emph{ternary} words respectively where $\Sigma_2 = \{a,b\}$ and $\Sigma_3 = \{a,b,c\}$.
We denote the class of regular languages, context-free languages, and context-sensitive languages by $\REG, \CFL,$ and $\CSL$ respectively.
These classes of languages are standard, and we will assume the reader has familiarity with them.
We denote the class of indexed languages and linear indexed languages by $\IL$ and $\LIL$ respectively.
Definitions of the grammars which generate indexed languages and linear indexed languages will be given in Section~\ref{sec:ins}.

Given a word $w \in \Sigma^*$ and an integer $p > 0$ we consider the possible outcomes of iteratively deleting $p$th-powers from $w$ until we have a $p$th-power-free word.
In particular, we are interested in when we can obtain the empty word $\eps$.
Let us now consider an example.

\begin{ex}
Let $w = ababbcbc \in \Sigma_3^*$ and $p = 2$.
So, we are considering squares in a ternary word.
The word $w$ has the squares $(ab)^2$, $b^2$, and $(bc)^2$.
Squares can be deleted from $w$ in the following ways
$$ (abab)bcbc \to bcbc \to \eps$$
$$ aba(bb)cbc \to abacbc$$
$$ abab(bcbc) \to abab \to \eps$$
each time ending with a square-free word.
Notice the resulting square-free word depends on how we choose to delete squares from $w$.
\label{ex:del}
\end{ex}

We call the operation $ux^pv \to uv$ a \emph{$p$-deletion from $ux^pv$ to $uv$}.
Now consider a sequence of $p$-deletions $(d_1, d_2, \cdots, d_{\ell})$ such that $d_i$ is a $p$-deletion from $w_{i-1}$ to $w_i$.
Such a sequence of $p$-deletions will be said to start with $w_0$ and terminate with $w_{\ell}$.
We will only consider sequences of $p$-deletions of this form.
A word $w$ is called \emph{$p$-deletable} if there exists some sequence of $p$-deletions which starts with $w$ and terminates with $\eps$.
A word $w$ is called \emph{$p$-strongly-deletable} if $w$ is $p$-deletable and the only $p$th-power-free word that can be obtain from $w$ by a sequence of $p$-deletions is $\eps$.
So, the word $w = ababbcbc \in \Sigma_3^*$ in Example~\ref{ex:del} is $2$-deletable but not $2$-strongly-deletable.
We allow the empty sequence of $p$-deletions and so $\eps$ is $p$-deletable and $p$-strongly-deletable for any $p > 0$.

Given any $k,p > 0$ we define the following languages
\begin{align*}
D_{k,p} &= \{w \in \Sigma_k^* : w \mathrm{\;is\;}p\mathrm{\mhyphen deletable}\}\\
SD_{k,p} &= \{w \in \Sigma_k^* : w \mathrm{\;is\;}p\mathrm{\mhyphen strongly \mhyphen deletable}\}
\end{align*}
consisting of $p$-deletable and $p$-strongly-deletable words over $\Sigma_k$ respectively.
Our focus will be on studying the languages $D_{k,p}$ and $SD_{k,p}$.

The study of powers in words has a long history.
Powers in words were first systematically studied by Thue~\cite{Thue, Thue1} where an interesting dichotomy is observed.
Every binary word of length at least $4$ must contain a square.
This can be seen by simply listing all $16$ binary words of length $4$.
However, Thue constructs an infinite binary word which is cube-free.
An infinite ternary word which is square-free is also constructed.
In our study we will seen a similar phenomenon where the behavior of binary words differs from ternary words and the behavior of squares differs from cubes.

We will focus on squares in Section~\ref{sec:reg}.
In Theorem~\ref{thm:binsq}, we show that $SD_{2,2} = D_{2,2}$ is a regular language while in Theorem~\ref{thm:notreg} we show that neither $SD_{k,2}$ or $D_{k,2}$ is a regular language for $k > 2$.
In Section~\ref{sec:ins} we show how our languages $D_{k,p}$ are related to Kari's theory of insertion~\cite{Kari}.
We show in Corollary~\ref{cor:index} that $D_{k,p}$ is an indexed language and $D_{k,2}$ is a linear indexed language for any $k,p > 0$.

We now make some basic observations about the languages $SD_{k,p}$ and $D_{k,p}$.
Take any $k,p > 0$ and $w \in \Sigma^*_k$.
First note we can check if $w$ is a $p$th-power in $O(|w|)$ time.
We can check if $w$ is $p$th-power free in polynomial time since $w$ has only $O(|w|^2)$ subwords.
Also, the length of any sequence of $p$-deletions starting with $w$ is $O(|w|)$.
Therefore we can verify if sequence of $p$-deletions results in $\eps$ in polynomial time, and we can also verify if a sequence of $p$-deletions results in a non-empty $p$th-power-free word in polynomial time.
It follows that for any $k,p > 0$ we have $D_{k,p} \in \NP$ and $SD_{k,p} \in co$-$\NP$.

Our next observation is that 
$$SD_{k,p} \subseteq D_{k,p}$$
which is immediate from the definitions.
In general this containment is strict as demonstrated by Example~\ref{ex:del}, but we will see equality of $D_{k,p}$ and $SD_{k,p}$ in some special cases in Section~\ref{sec:reg}.
The special cases where equality occurs are the trivial cases of deleting $1$st-powers and deleting powers over $\Sigma^*_1$ as well as the case of deleting squares in binary words.
Next we give a lemma which contains a necessary condition for a word to be $p$-deletable.

\begin{lem}
If $w \in D_{k,p}$, then $|w|_a \equiv 0 \pmod p$ for all $a \in \Sigma_k$.
In particular if $w \in D_{k,p}$, then $|w| \equiv 0 \pmod p$.
\label{lem:modp}
\end{lem}

\begin{proof}
For the empty word $\eps$ we have $|\eps|_a = 0$ for all $a \in \Sigma_k$.
Since the number of occurrences of any letter in a word is preserved modulo $p$ when preforming a $p$-deletion, the result follows.
\end{proof}

We can use Lemma~\ref{lem:modp} to determine that a word is not $p$-deletable.
Consider the following example.

\begin{ex}[Fibonacci words]
Fix $p>0$. The Fibonacci words are words over $\Sigma_2= \{a,b\}$ defined by $S_0 = a$, $S_1 = ab$, and $S_n = S_{n-2}S_{n-1}$ for $n \geq 2$.
Observe that $|S_n|_a = F_{n+1}$ and $|S_n|_b = F_n$ where $F_n$ denotes the $n$th Fibonacci number.
Since any two consecutive Fibonacci numbers are relatively prime, we can never have $|S_n|_a \equiv 0 \pmod p$ and $|S_n|_b \equiv 0 \pmod p$ simultaneously.
Thus, by Lemma~\ref{lem:modp} it follows that $S_n \not\in D_{2,p}$ for any $n \geq 0$.
\end{ex}

\section{Squares and Regular Languages}
\label{sec:reg}
In this section we will given an explicit description of the languages $D_{k,p}$ and $SD_{k,p}$ for certain values of $k,p > 0$.
For certain values of $k,p>0$ for which we can describe the languages $D_{k,p}$ and $SD_{k,p}$, these two languages turn out to be equal and are regular languages.
We also show in this section to $D_{k,2}$ and $SD_{k,2}$ are not regular for $k > 2$.
Recall that $\REG$ denotes the class of regular languages.

We first consider two trivial cases, for any $k,p > 0$
\begin{align*}
D_{k,1} &= SD_{k,1} = \Sigma^*_k\\
D_{1,p} &= SD_{1,p} = \{w \in \Sigma_1^* : p \mathrm{\;divides\;} |w| \}
\end{align*}
Notice in both of the above cases, we have equality of the language of deletable words and the language of strongly-deletable words.
Also, both $D_{k,1} = SD_{k,1}$ and $D_{1,p} = SD_{1,p}$ are regular languages.
The first nontrivial case we encounter is $k=2$ and $p=2$ where we look at squares in binary words.
Here, we still have equality of the languages of deletable and strongly-deletable words, and this language again is regular.

\begin{thm}
Let $\Sigma_2 = \{a, b\}$, then 
$$SD_{2,2} = D_{2,2} = \{w \in \Sigma_2^* : |w|_a \mathrm{\;and\;} |w|_b \mathrm{\;are\;even}\}.$$
\label{thm:binsq}
\end{thm}

\begin{proof}
By Lemma~\ref{lem:modp} both $|w|_a$ and $|w|_b$ being even is a necessary condition for a word to be $2$-deletable, and hence also for a word to be $2$-strongly-deletable.
Recall that any binary word $w$ with $|w| \geq 4$  is not square-free.
Now consider any binary word $w$ with both $|w|_a$ and $|w|_b$ even, and so in particular $|w|$ is even.
Arbitrarily delete squares to obtain $w'$ with $|w'| < 4$.
We will have $|w'| = 0$ or $|w'| = 2$ since $|w|$ was even.
Moreover, since both $|w|_a$ and $|w|_b$ are even we must have $w' = \eps$, $w' = aa$, or $w' = bb$.
Therefore $w$ is 2-deletable.
Since we deleted squares arbitrarily $w$ is in fact $2$-strongly-deletable.
\end{proof}

So far, in all the cases we have looked at, the necessary condition in Lemma~\ref{lem:modp} has turned out to also be a sufficient condition for a word to be $p$-deletable.
This condition is not always sufficient.
For example, note that $w = abacbc \not\in D_{3,2}$ even though $|w|_a$, $|w|_b$, and $|w|_c$ are all even.
As we continue looking at squares, but now over a larger alphabet, the techniques applied to squares in binary words can no longer be used due to the existence of arbitrarily long square-free ternary words.
We will see in what follows that the existence of an infinite square-free ternary word causes the languages $SD_{k,2}$ and $D_{k,2}$ to non-regular for $k > 2$.
For any $k>0$ and word $w = a_1a_2 \cdots a_n \in \Sigma^*_k$ we define the \emph{reverse of $w$} by $w^{\R} := a_n a_{n-1} \cdots a_1$.
We remark the for any $u \in \Sigma^*_k$ and any $p>0$ that $(u^p)^{\R} = (u^{\R})^p$.
It follows that $w \in D_{k,p}$ if and only if $w^{\R} \in D_{k,p}$, and similarly $w \in SD_{k,p}$ if and only if $w^{\R} \in SD_{k,p}$.
Also, $(w^{\R})^{\R} = w$ and $w$ is $p$th-power-free if and only if $w^{\R}$ is $p$th-power-free.

\begin{lem}
If $x \in \Sigma_k^*$ is square-free, then $xx^{\R} \in SD_{k,2}$.
\label{lem:SD}
\end{lem}
\begin{proof}
Let $x = a_1a_2\cdots a_n \in \Sigma_k^*$ be square-free.
First note for any $x \in \Sigma_k^*$ we have $xx^R = a_1a_2 \cdots a_{n-1}a_n a_n a_{n-1} \cdots a_2 a_1 \in D_{k,2}$ by inductively repeating the deletion
$$a_1a_2 \cdots a_{n-1}(a_n a_n) a_{n-1} \cdots a_2 a_1 \to a_1a_2 \cdots a_{n-1} a_{n-1} \cdots a_2 a_1.$$

We will show that $xx^{\R} \in SD_{k,p}$ by induction.
Observe that if $|x| = 0$, then $xx^{\R} = \eps \in SD_{k,2}$.
To show $xx^{\R} \in SD_{k,p}$ it suffices to show that after deleting any square from $xx^{\R}$ we obtain a $2$-strongly-deletable word.
Consider any square in $xx^R$ which must be of the form
$$u^2 = a_i a_{i+1} \cdots a_n a_n a_{n-1} \cdots a_j$$
since it must cross from $x$ into $x^R$ as both $x$ and $x^R$ are square-free.
We claim that we must have $i = j$.
Assume $i < j$ then we have some $i < \ell < n$ where $u = a_i a_{i+1} \cdots a_{\ell}$ and $u = a_{\ell+1} \cdots a_n a_n \cdots a_j$.
We then see that the square $a_n^2$ must occur in $x$.
This is a contradiction to $x$ being square-free.
A similar contradiction is reached if $i > j$.
So, $i = j$ and after deleting our square we obtain the word $a_1 \cdots a_{i-1} a_{i-1} \cdots a_1$ which is $2$-strongly-deletable by induction.
Therefore it follows that $xx^R \in SD_{k,2}$.
\end{proof}

\begin{lem}
If $x,y \in \Sigma_k^*$ such that $xy$ is square-free and $|y| > 0$, then $xyx^{\R} \not\in D_{k,2}$.
\label{lem:D}
\end{lem}
\begin{proof}
We will show $xyx^{\R} \not\in D_{k,2}$ by showing that if we delete any square from $xyx^{\R}$ we obtain a word which is not $2$-deletable.
If $|x| = 0$ we are done since $xyx^{\R} = y \not\in D_{k,2}$ as $y$ is square-free and $|y| > 0$.
We induct on $|x|$.
We can assume that $|y| = 2n$, otherwise $xyx^{\R} \not\in D_{k,2}$ by Lemma~\ref{lem:modp} since $|xyx^{\R}|$ would be odd.
Let $x = a_1a_2 \cdots a_m$ and $y = b_1b_2 \cdots b_{2n}$.
We note any square $u^2$ in $xyx^{\R}$ must cross from $xy$ into $x^{\R}$ since $xy$ and $x^{\R}$ are square-free. We consider the following cases for how this square can occur.

First if $u^2 = b_i \cdots b_{2n} a_m \cdots a_j$, then after deleting $u^2$ we obtain
$$w' = (a_1 \cdots a_{j-1})(a_j \cdots a_m b_1 \cdots b_{i-1})(a_{j-1} \cdots a_1).$$
Let $x' = a_1 \cdots a_{j-1}$ and $y' = a_j \cdots a_m b_1 \cdots b_{i-1}$, then $x'y'$ is square-free with $|y'| > 0$ and $w' = x'y'(x')^{\R}$ where $|x'| < |x|$.
So, $w' \not\in D_{k,2}$ by induction

Second if $u^2 = a_i \cdots a_m y a_m \cdots a_j$ we have three subcases.
If $i > j$ then after deleting $u^2$ 
$$w' = a_1 \cdots a_{i-1} a_{j-1} \cdots a_1 = (a_1 \cdots a_{j-1})(a_j \cdots a_{i-1})(a_{j-1} \cdots a_1).$$
Let $x' =a_1 \cdots a_{j-1}$ and $y' =a_j \cdots a_{i-1}$, then $x'y'$ is square-free with $|y'| > 0$ and $w' = x'y'(x')^R$ where $|x'| < |x|$.
So,  $w' \not\in D_{k,2}$ by induction.
If $i < j$ then after deleting $u^2$ we have
$$w' =a_1 \cdots a_{i-1} a_{j-1} \cdots a_1 = (a_1 \cdots a_{i-1})(a_{j-1} \cdots a_i)(a_{i-1} \cdots a_1).$$
Let 
$$w'' = (w')^R =  (a_1 \cdots a_{i-1})(a_i \cdots a_{j-1})(a_{i-1} \cdots a_1)$$
and also let $x'' = a_1 \cdots a_{i-1}$, and $y'' = a_i \cdots a_{j-1}$.
Now $x''y''$ is square-free with $|y''| > 0$ and $w'' = x''y''(x'')^{\R}$ where $|x''| < |x|$.
So, $w'' \not\in D_{k,2}$ by induction and hence $w' = (w'')^{\R} \not\in D_{k,2}$.
The final case is $i = j$ which we claim cannot happen.
If it were the case that $u^2 = a_i \cdots a_m y a_m \cdots a_i$, then since $y =b_1 \cdots b_{2n}$ we would have 
$$a_i \cdots a_m b_1 \cdots b_n = b_{n+1} \cdots b_{2n} a_m \cdots a_i.$$
This would imply a contradiction to $xy$ being square-free since we would have $b_n = a_i$ and $b_{n+1} = a_i$ thus $a_i^2 = b_nb_{n+1}$ would be a square contained in $xy$.
\end{proof}

The proof of the next theorem uses the Myhill-Nerode Theorem which provides a necessary and sufficient condition for a language to be regular. 
Given a language $L$ over $\Sigma$, we consider the equivalence relation $\sim_L$ on $\Sigma^*$ defined as follows. 
For any $x, y \in \Sigma^*$ set $x \sim_L y$ whenever for every $z \in \Sigma^*$ we have $xz \in L$ if and only if $yz \in L$.
The Myhill-Nerode Theorem says a language $L$ is regular if and only if the equivalence relation $\sim_L$ has a finite number of equivalence classes. 

\begin{thm}
For $k > 2$,  $SD_{k,2} \not\in \REG$ and $D_{k,2} \not\in \REG$.
\label{thm:notreg}
\end{thm}
\begin{proof}
Let $L = SD_{k,2}$ or $L = D_{k,2}$ for some $k > 2$ and consider an infinite square-free word $w = a_1a_2 \cdots$ over $\Sigma_k^*$.
Note an infinite square-free word exists whenever $k > 2$.
Let $w_n = a_1a_2 \cdots a_n$ be the first $n$ letters of $w$.
Consider $m < n.$ Then using Lemma~\ref{lem:SD} and Lemma~\ref{lem:D} we have $w_m w_m^R \in L$ but $w_nw_m^R = w_m (a_{m+1} \cdots a_n)w_m^R \not\in L$.
Thus $w_m \not\sim_L w_n$ and $L$ is not regular since $\sim_L$ has infinitely many equivalence classes.
\end{proof}

\section{Insertion and Indexed Languages}
\label{sec:ins}
Given two languages $L_1$ and $L_2$ we get a new language $(L_1 \ins L_2)$ called the \emph{insertion} of $L_2$ into $L_1$   defined by
$$(L_1 \ins L_2) := \{xyz : xz \in L_1, y \in L_2\}.$$
Insertion can be iterated  by letting $(L_1 \ins^0 L_2) = L_1$ and defining 
$$(L_1 \ins^i L_2) := ((L_1 \ins^{i-1} L_2) \ins L_2)$$
for $i > 0$.
We will be concerned with
$$(L_1 \ins^* L_2) := \bigcup_{i \geq 0} (L_1 \ins^i L_2).$$
This notion of insertion is defined and studied by Kari in~\cite{Kari}.

We now show that our language $D_{k,p}$ can be described in terms of insertions.
We have defined $D_{k,p}$ in terms of $p$-deletions so that we think of $D_{k,p}$ as the words that we can reduce to $\eps$ with a sequence of $p$-deletions.
The next lemma says that we can equivalently think of $D_{k,p}$ as those words which can be built from $\eps$ by insertions of $p$th-powers.
For $k,p > 0$ we let $L_{k,p} = \{w^p : w \in \Sigma^*_k\}$ denote the language of $p$th-powers over $\Sigma_k$.

\begin{lem}
For any $k,p >0$, $D_{k,p}  = (\eps \ins^* L_{k,p})$.
\label{lem:ins}
\end{lem}

\begin{proof}
We first show that $(\eps \ins^* L_{k,p}) \subseteq D_{k,p}$ by showing that  $(\eps \ins^i L_{k,p}) \subseteq D_{k,p}$ for all $i \geq 0$.
We have $\{\eps\} = (\eps \ins^0 L_{k,p}) \subseteq D_{k,p}$ and proceed by induction.
Assume $(\eps \ins^i L_{k,p}) \subseteq D_{k,p}$ and take $w \in (\eps \ins^{i+1} L_{k,p})$ so $w = xyz$ for $y \in  L_{k,p}$ and $xz \in (\eps \ins^i L_{k,p})$.
Then the $p$-deletion $w = xyz \to xz \in D_{k,p}$ shows that $w \in D_{k,p}$ as desired.

Conversely take $w \in D_{k,p}$ and let
$$w = w_0 \to w_1 \to \cdots \to w_{\ell} = \eps$$
be a sequence of $p$-deletions resulting in $\eps$.
Reading this sequence in reverse provides a sequence of insertions of $p$th-powers which shows that $w \in (\eps \ins^{\ell} L_{k,p}) \subseteq (\eps \ins^* L_{k,p})$.
\end{proof}

We will now give the definitions of indexed grammars and linear indexed grammars.
These grammars can generate languages between $\CFL$ and $\CSL.$
An indexed grammar is essentially a context-free grammar with the addition that each nonterminal symbol in a production rule receives a stack.
Indexed grammars were introduced by Aho~\cite{Aho}.
We give a formal definition below following Hopcroft and Ullman~\cite{HU}.

\begin{defn}
An indexed grammar is a 5-tuple $(N,T,I,P,S)$ where $N$ is the set of non-terminals, $T$ is the set of terminals, $I$ is the set of indices, $P$ is the finite set of productions, and $S \in N$ is the start symbol. 
Each production rule must be of one of the following forms:
\begin{align*}
A[\sigma] &\to \alpha[\sigma]\\
A[\sigma] &\to B[\gamma \sigma]\\
A[\gamma \sigma] &\to \alpha[\sigma]
\end{align*}
In the production rules above $A, B \in N,$ $\gamma \in I,$ $\sigma \in I^*,$ and $\alpha \in (N \cup T)^*$.
The notation $\alpha[\sigma]$ means each non-terminal symbol in $\alpha$ receives the stack $[\sigma]$.
\label{def:IG}
\end{defn}

Notice that the second type of production rule in Definition~\ref{def:IG} can be thought as a ``push,'' while the third type of production rule can be thought of as a ``pop.''
This motivates the terminology where $[\sigma]$ is referred to as the ``stack.''
Also note how the presence of the stack makes indexed grammars differ from context-free grammars.
For $A \in N$ and $\alpha \in (N \cup T)^*$ a production rule of the form $A[\sigma] \to \alpha[\sigma]$ is really an infinite family of production rules with one production rule for each $\sigma \in I^*.$ 
It turns out that the language $L_{k,p}$ consisting of $p$th-powers over $\Sigma_k$ is generated by an indexed grammar.
We now give an indexed grammar for $L_{k,p}$.

\begin{ex}
Here we give an indexed grammar for $L_{k,p}$ for any $k,p > 0$.
Consider the grammar $G = (N,\Sigma_k,I,P,S)$ where $N = \{S,T\}$, $\Sigma_k = \{a_i : 1 \leq i \leq k\}$, $I = \{\gamma_i : 1 \leq i \leq k\}$ and the production rules in $P$ are:
\begin{align*}
S [\sigma] &\to S[\gamma_i \sigma]  & 1 \leq i \leq k\\
S[\sigma] &\to (T[\sigma])^p\\
T[\gamma_i \sigma] &\to a_i T[\sigma] & 1 \leq i \leq k\\
T[]  &\to \eps
\end{align*}
\label{ex:pg}
\end{ex}

We now define a linear indexed grammar which is similar to an indexed grammar, but has the restriction that only one nonterminal symbol can receive the stack per production rule.
Linear indexed grammars were proposed by Gazdar~\cite{Gazdar}.
These grammars are \emph{mildly context-sensitive} and are weakly equivalent to many other grammars including tree adjoin grammars, head grammars, and combinatory categorial grammars~\cite{VSW1}.
Linear indexed grammars can be parsed in polynomial time~\cite{VSW}.

\begin{defn}
A linear indexed grammar is a 5-tuple $(N,T,I,P,S)$ where $N$ is the set of non-terminals, $T$ is the set of terminals, $I$ is the set of indices, $P$ is the finite set of productions, and $S \in N$ is the start symbol. Each production rule of one of the following forms:
\begin{align*}
A[\sigma] &\to \alpha[] B[\sigma] \beta[]\\
A[\sigma] &\to \alpha[] B[\gamma \sigma] \beta[]\\
A[\gamma \sigma] &\to \alpha[] B[\sigma] \beta[]\\
A[] &\to w
\end{align*}
In the production rules above $A, B \in N,$ $\gamma \in I,$ $\sigma \in I^*,$ $\alpha \in (N \cup T)^*,$ and $w \in T^*$.
\end{defn}

We have seen in Example~\ref{ex:pg} that $L_{k,p}$ is generated by an indexed grammar for any $k,p > 0$.
For $p=2,$ the language of squares $L_{k,2}$ can be generated by a linear indexed grammar.
We now given an linear indexed grammar for $L_{k,2}$.

\begin{ex}
Here we give a linear indexed grammar for $L_{k,2}$ for any $k > 0$.
Consider the grammar $G = (N,\Sigma_k,I,P,S)$ where $N = \{S,T\}$, $\Sigma_k = \{a_i : 1 \leq i \leq k\}$, $I = \{\gamma_i : 1 \leq i \leq k\}$ and the production rules in $P$ are given by:
\begin{align*}
S[\sigma] &\to a_i S[\gamma_i \sigma] & 1 \leq i \leq k\\
S[\sigma] &\to T[\sigma]\\
T[\gamma_i \sigma] &\to T[\sigma] a_i & 1 \leq i \leq k\\
T[] &\to \eps
\end{align*}
\label{ex:sqg}
\end{ex}

A language is call an \emph{indexed language} if it can be generated by an indexed grammar.
Similarly, a language is call a \emph{linear indexed language} if it can be generated by a linear indexed grammar.
Recall that $\IL$ and $\LIL$ denote the class of indexed languages and the class of linear indexed languages respectively.
Also recall that the class of context-free languages is denoted by $\CFL$ while the class of context-sensitive languages is denoted by $\CSL.$
These classes of languages satisfy the following strict inclusions
$$\CFL \subsetneq \LIL \subsetneq \IL \subsetneq \CSL.$$
We now prove a theorem which shows that indexed languages and linear indexed languages are closed under iterated insertion.
We call the readers attention to~\cite[Theorem 2.6] {Kari} and~\cite[Theorem 2.7]{Kari} which give the analogous results for context-free languages and context-sensitive languages.

\begin{thm}
If $L_1$ and $L_2$ are indexed languages, then $(L_1 \ins^* L_2)$ is an indexed language.
Also if $L_1$ and $L_2$ are linear indexed languages, then $(L_1 \ins^* L_2)$ is a linear indexed language.
\label{thm:ind}
\end{thm}
\begin{proof}
We first prove the case were $L_1$ and $L_2$ are indexed languages.
Let $L_i$ be generated by the indexed grammar $G_i = (N_i,T_i,I_i,P_i,S_i)$ with $i \in \{1,2\}$.
Now consider the indexed grammar $G = (N,T,I,P,S_1)$ where $N = N_1 \cup N_2 \cup \{S'\}$, $T= T_1 \cup T_2$, and $P = P_1 \cup P_2 \cup P'$.
The rules in $P'$ will be the rules which insert $L_2$.
This will be achieved by inserting the symbol $S'$.
For each rule in $P_1 \cup P_2$ of the form $A[\sigma] \to \alpha[\sigma]$ or $A[\gamma \sigma] \to \alpha[\sigma]$, we have a rule in $P'$ of the form $A[\sigma] \to \alpha'[\sigma]$ or $A[\gamma \sigma] \to \alpha'[\sigma]$ where $\alpha'$ is obtained from $\alpha$ by inserting $S'$ before or after any terminal symbol.
We also have the rules $S'[\gamma \sigma] \to S'[\sigma]$ for any $\gamma \in I$ and $\sigma \in I^*$, and we have the rule $S'[] \to S_2[]$.
This indexed grammar $G$ defines the language $(L_1 \ins^* L_2)$.

The case were $L_1$ and $L_2$ are linear indexed languages is similar.
When building a linear indexed grammar for $(L_1 \ins^* L_2)$ we proceed exactly as above, but we do not actually need the symbol $S'$.
We can simply insert $S_2$ directly with an empty stack.
\end{proof}

By combining Theorem~\ref{thm:ind} with Lemma~\ref{lem:ins}, Example~\ref{ex:pg}, and Example~\ref{ex:sqg} we get the following corollary.

\begin{cor}
For any $k,p>0$ we have $D_{k,p} \in \IL$ and $D_{k,2} \in \LIL$.
\label{cor:index}
\end{cor}

We have initiated the study of deleting powers in words and believe that this topic is interesting from both a combinatorics on words perspective as well as from a formal language perspective.
There are some natural open questions involving the study of the languages $D_{k,p}$ and $SD_{k,p}$ which we will briefly outline.
We have shown in Theorem~\ref{thm:notreg} that $D_{k,2}$ and $SD_{k,2}$ are not regular languages for $k > 2.$
We have also shown in Corollary~\ref{cor:index} that $D_{k,p}$ is an indexed language for any $k,p > 0$ while $D_{k,2}$ is a linear indexed language for any $k > 0$.
However, outside of the few special cases in Section~\ref{sec:reg} where some languages are shown to be regular, we do not have any proof showing whether or not these languages are context-free.
Furthermore, we do not have any results which determine if $D_{k,p}$ is a linear indexed language for $p > 2.$
The corresponding questions for the languages of strongly-deletable words seem more difficult.
A nontrivial result analogous to Corollary~\ref{cor:index} for the languages $SD_{k,p}$ is desirable but not known to us at this time.

The smallest open cases are squares in ternary words and cubes in binary words.
For squares in ternary words we know that $D_{3,2} \in \LIL.$
A next step would be to determine if $D_{3,2}$ is context-free or not, and perhaps a proof for $D_{3,2}$ would extend to larger alphabets.
For cubes in binary words is known that $L_{2,3} \not\in \LIL$ (see~\cite[Lemma 4.15]{Kall} for instance), but we do not know the whether $D_{2,3} \in \LIL$ or $D_{2,3} \not\in \LIL.$
The fact that the language of binary cubes is not a linear indexed language suggests $D_{2,3}$ may not be a linear indexed language.
However, $L_{k,p}$ and $D_{k,p}$ can certainly behave differently.
Note that the language of binary squares $L_{2,2}$ is not regular, in fact $L_{2,2}$ is not context-free, but $D_{2,2} = SD_{2,2}$ is a regular language.

\section*{Acknowledgement}
The author thanks the anonymous referees for their suggestions which have improved this paper.

\bibliography{Powers}

\end{document}